\documentclass[lineno]{acmart}

\AtBeginDocument{%
  }

\setcopyright{acmlicensed}
\copyrightyear{2018}
\acmYear{2018}
\acmDOI{XXXXXXX.XXXXXXX}
\acmConference[Conference acronym 'XX]{Make sure to enter the correct
  conference title from your rights confirmation email}{June 03--05,
  2018}{Woodstock, NY}
\acmISBN{978-1-4503-XXXX-X/2018/06}

\usepackage{algorithm}  
\usepackage{algorithmic}  
\usepackage{moreverb}
\usepackage{subfig}
\usepackage{xcolor}
\usepackage{float}
\usepackage{notoccite}

\begin{document}

\title{Crypto-Assisted Graph Degree Sequence Release under Local Differential Privacy}

\author{Xiaojian Zhang}
\email{xjzhang82@alu.ruc.edu.cn}
\affiliation{%
  \institution{Henan University of Economics and Law}
  \city{Zhengzhou}
  \state{Henan}
  \country{China}
}

\author{Junqing Wang}
\affiliation{%
  \institution{Guangzhou University}
  \city{Guangzhou}
  \country{China}}
\email{csjunqing@e.gzhu.edu.cn}

\author{Kerui Chen}
\affiliation{%
  \institution{Henan University of Economics and Law}
  \city{Zhengzhou}
  \country{China}}

\author{Peiyuan Zhao}
\affiliation{%
 \institution{Henan University of Economics and Law}
 \city{Zhengzhou}
 \state{Henan}
 \country{China}}

\author{Huiyuan Bai}
\affiliation{%
 \institution{Henan University of Economics and Law}
 \city{Zhengzhou}
 \state{Henan}
 \country{China}}




\renewcommand{\shortauthors}{Xiaojian Zhang et al.}

\begin{abstract}
Given a graph $G$ defined in a domain $\mathcal{G}$, we investigate locally differentially private mechanisms to release a degree sequence on $\mathcal{G}$ that accurately approximates the actual degree distribution. Existing solutions for this problem mostly use graph projection techniques based on edge deletion process, using a threshold parameter $\theta$ to bound node degrees. However, this approach presents a fundamental trade-off in threshold parameter selection. While large $\theta$ values introduce substantial noise in the released degree sequence, small $\theta$ values result in more edges removed than necessary. Furthermore, $\theta$ selection leads to an excessive communication cost. To remedy existing solutions' deficiencies, we present CADR-LDP, an efficient framework incorporating encryption techniques and differentially private mechanisms to release the degree sequence. In CADR-LDP, we first use the crypto-assisted Optimal-$\theta$-Selection method to select the optimal parameter with a low communication cost. Then, we use the LPEA-LOW method to add some edges for each node with the edge addition process in local projection. LPEA-LOW prioritizes the projection with low-degree nodes, which can retain more edges for such nodes and reduce the projection error. Theoretical analysis shows that CADR-LDP satisfies $\epsilon$-node local differential privacy. The experimental results on eight graph datasets show that our solution outperforms existing methods.
\end{abstract}



\received{20 February 2007}
\received[revised]{12 March 2009}
\received[accepted]{5 June 2009}

\maketitle

\section{Introduction}
In graph data, the degree sequence aims to describe the degree probability distribution, which provides insight into the structure and properties of the graph. 
However, the release of the graph degree sequence is carried out on sensitive graph data, which could be leaked through the publication results \cite{cite1, cite2, cite3}. Thus, models that can release the degree sequence while still preserving the privacy of individuals in the graph are needed. A promising model is local differential privacy (LDP)\cite{cite4, cite5}, where each user locally encodes and perturbs their data before submitting it to a collector. This model eliminates the need for a trusted collector, empowering users to retain control over their actual degree information. Two natural LDP variants are particularly suited for graph data: edge-LDP \cite{cite6} and node-LDP \cite{cite7}. Intuitively, the former protects relationships among nodes, and the latter provides a stronger privacy guarantee by protecting each individual and their associated relationships. While node-LDP offers a more robust privacy guarantee than edge-LDP, achieving it is more challenging. This is because, in the worst case, removing a single node can impact up to $n-1$ other nodes (where $n$ is the total number of nodes), which may lead to high sensitivity of degree sequence release, particularly in large graphs. If the Laplace mechanism \cite{cite8, cite9}(e.g., adding noise sampled from Lap($\frac{n-1}{\epsilon}$)) is applied to protect degree counts, the resulting perturbation may severely distort the true values.
For this reason, most existing methods \cite{cite6, cite7, cite10, cite11, cite12, cite13, cite14, cite15, cite16, cite17} adopt edge-LDP to protect the publication of graph data. 

To address the high sensitivity, a key technique to satisfy the node-LDP is that of $\text{\tt{graph projection}}$, which projects a graph into a $\theta$-degree-bounded graph whose maximum degree is no more than $\theta$ \cite{cite2, cite18, cite19, cite20}. Current studies often employ edge deletion or edge addition processes to limit the degree count of each node \cite{cite18, cite19, cite20}. A critical challenge is to preserve as much information about edges as possible in the projection process while releasing the degree sequence from a projected graph. However, existing projection solutions \cite{cite2, cite18} are only proposed for the central differential privacy model, which cannot be used to edge-LDP directly. This is because in the central setting, the collector can easily estimate the threshold value $\theta$ in terms of the global view, while in the local setting, since each user can only see its degree information, but not its neighboring information. It is challenging to estimate $\theta$ about the entire graph. Existing studies (e.g. $\text{\tt{EDGE-REMOVE}}$ \cite{cite19, cite20}) often rely on an edge deletion process to release the degree sequence with edge-LDP. Given the parameter $\theta$, in
the edge deletion process, each user $u_i$ randomly samples $d_i-\theta$ edges among its neighboring nodes to be deleted, where $d_i$ denotes the degree of $u_i$. The edge deletion process, however, removes significantly
more edges than necessary. Three limitations in the solutions based on the edge deletion process have yet to be solved: (1) it is difficult to obtain the optimal parameter $\theta$. A larger $\theta$ causes noise with larger magnitude to be added, while a lower $\theta$ leads to more edges being pruned; (2) The existing solutions do not consider the willingness of the users whose edges are deleted, resulting in excessive edge deletions. And (3) those methods fail to account for the ordering of degree counts among their neighboring nodes. To illustrate these limitations, Figure 1 shows an example using the edge deletion process. 
\vspace{-8pt}
\begin{figure}[htbp]
  \centering
  \includegraphics[scale=1.4]{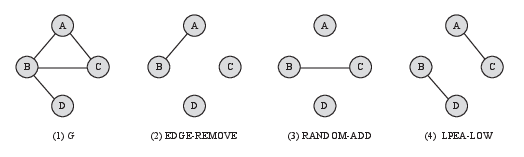}
  \caption{Comparison of EDGE-REMOVE, RANDOM-ADD, and LPEA-LOW}
\end{figure}

\textit{Example 1. Given $\theta = 1$, the degree of node B is 3. We need to delete two edges randomly from the set $\{$BA, BC, BD$\}$. If edges BC and BD are removed, the bounded graph is shown as Fig. 1(2). Node D’s degree reaches zero after removing BD. Nevertheless, given its single edge (BD), node D would resist this deletion based on its edge deletion willingness.}

Unlike existing node-LDP approaches that rely on edge deletion projection for degree sequence release, our analysis shows that random edge addition projection (e.g., $\text{\tt{RANDOM-ADD}}$ \cite{cite2}) achieves higher accuracy than edge deletion. Based on the Facebook dataset, we use the mean absolute error (denoted as MSE) and $\frac{|E'|}{|E|}$ ratio to measure the accuracy of the above two processes, where $E'$ and $E$ represent the edge sets of the projected graph and original graph, respectively. As shown in Fig. 2, when the projection parameter $\theta$ varies from 1 to 50, the solutions based on the edge addition process demonstrate superior performance over the edge deletion-based methods in terms of both MAE and $\frac{|E'|}{|E|}$ ratio.

\begin{figure}[htbp]
  \centering
  \includegraphics[scale=1.5]{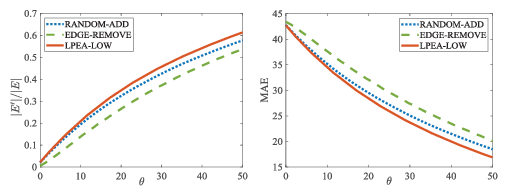}
  \caption{Comparison of EDGE-REMOVE, RANDOM-ADD, and LPEA-LOW on Facebook dataset}
\end{figure}

Although solutions with the edge addition process can enhance degree distribution accuracy, they do not consider the order of edge additions; instead, they perform them randomly. Our analysis reveals that initiating edge additions from low-degree nodes preserves more functional edges in the projected graph. Based on this analysis, we propose $\text{\tt{LPEA-LOW}}$, a local edge addition projection method that prioritizes low-degree nodes and sequentially adds edges until $\theta$ is met. Fig. 2 shows that $\text{\tt{LPEA-LOW}}$ is better than $\text{\tt{RANDOM-ADD}}$ and $\text{\tt{EDGE-REMOVE}}$. Based on Fig. 1, we give the second example to demonstrate the advantages of $\text{\tt{LPEA-LOW}}$, shown as example 2.

\textit{Example 2. Given $\theta$ = 1, assuming the order in which nodes start local projection is: B$\rightarrow$C$\rightarrow$A$\rightarrow$D.  In $\text{\tt{EDGE-REMOVE}}$, if node B randomly deletes edges BC and BD to reach $\theta$ = 1, node B completes its projection. Node C then begins projection, but no deletion is needed since it only has one edge. Node A starts projection and randomly deletes edge AC to reach $\theta$ = 1, after which node A completes its projection. Node D has no edges, so no deletion is required. The result of $\text{\tt{EDGE-REMOVE}}$ projection is shown in Fig. 1(2), where only edge BA remains. In $\text{\tt{RANDOM-ADD}}$, if node B randomly adds edge BC, nodes B and C meet $\theta$ = 1 and complete their projections. Although nodes A and D do not meet the $\theta$=1 condition, edges BA, AC, and BD can no longer be added. The result of  $\text{\tt{RANDOM-ADD}}$ projection is shown in Fig. 1(3), where only edge BC remains. In $\text{\tt{LPEA-LOW}}$, node B does not randomly select edges to add but preferentially chooses to connect with its low-degree neighbor D. Similarly, node C preferentially connects with its lower-degree neighbor A. After meeting the $\theta$ = 1, the $\text{\tt{LPEA-LOW}}$ projection result is shown in Fig. 1(4). This example demonstrates that, under the same conditions, our method preserves more edges.}

Since node-LDP-based graph projection lacks a global view, the process depends on mutual negotiation between nodes and message exchange between nodes and the collector to determine the degree count ordering of each node's neighbors. However, direct negotiation and exchange will lead to several privacy risks: (1) the selection of $\theta$ may leak privacy, as the collector can only obtain the optimal $\theta$ by aggregating the actual degree counts of each node. (2) The edge deletion and edge addition processes may introduce a privacy leak, as these operations can reveal whether a node's degree is greater than $\theta$. And (3) the projected degree may leak privacy. To address these issues, $\text{\tt{EDGE-REMOVE}}$ uses the WRR \cite{cite21} mechanism, OPE \cite{cite22} encryption protocol, and secure aggregation technique to release the degree sequence under node-LDP. This method, however, may result in high computational and communication costs. This is because the aggregation in each round requires the local projection. Inspired by the edge addition process from the nodes with low degrees, which can improve the accuracy of releasing the degree sequence, we propose a crypto-assisted framework, called $\text{\tt{CADR-LDP}}$, to release the degree sequence. Our main contributions are threefold:
\begin{itemize}
\item To overcome the excessive edge loss in existing edge deletion-based graph projection methods, we propose $\text{\tt{LPEA-LOW}}$ that employs the edge addition process to bound the original graph. Our core idea lies in prioritizing edge additions for low-degree nodes, which is inspired by the fact that most graphs follow a long-tail degree distribution where low-degree nodes dominate. Thereby, our method can preserve more edges among these nodes. To prevent privacy leakage during the edge addition process, we incorporate the WRR, exponential, and Laplace mechanisms to protect node degree information.
\item To derive the optimal projection parameter $\theta$, we propose two crypt-assisted optimization methods, Optimal-$\theta$-Selection-by-Sum and Optimal-$\theta$-Selection-by-Deviation. These two approaches employ distinct optimization strategies: (1) empirical minimization of the summed error function, and (2) analytical derivation via error function differentiation, to determine $\theta$. Notably, the latter method achieves lower communication cost than the former.
\item Extensive experiments on real graph datasets demonstrate that our proposed methods can achieve a better trade-off between privacy and utility across distinct graph utility metrics.
\end{itemize}

\section{Preliminaries}
This section introduces the problem definition, node local differential privacy, and secure aggregation.

\subsection{Problem Definition}
This paper focuses on an undirected graph with no additional attributes on edges and nodes. Given an input graph $G = (V, E)$, where $V = \{v_1, v_2, ..., v_n\}$ denotes the set of nodes, and $E\subseteq V\times V$ denotes the set of edges, we want to release a degree sequence under node-LDP. Let $B_i = \{b_{i1}, b_{i2}, ..., b_{in}\}$ be the adjacent bit vector of the node $v_i$, where $b_{ij} \in \{0, 1\}$ indicates connectivity to the node. The degree $d_i$ of node $v_i$ is given by $d_i = \sum_{j=1}^{n}b_{ij}$. The collector aggregates a perturbed degree sequence $s = \{d_1, d_2, ..., d_n\}$ from each local user and releases the degree distribution $dist(G)$. We adopt two common utility metrics (e.g., mean squared error (MSE) and mean absolute error (MAE)) to evaluate the accuracy of our solutions. MSE and MAE are defined as follows:
\begin{equation}
  MSE(s, s') = \frac{1}{n}\sum_{i=1}^{n}(d_i - \widetilde{d_i})^2
\end{equation}

\begin{equation}
  MAE(s, s') = \frac{1}{n}\sum_{i=1}^{n}|d_i - \widetilde{d_i}|
\end{equation}

where $d_i$ and $\widetilde{d_i}$ denote the original degree and noise degree, respectively.

\subsection{Differential Privacy}
In the context of degree sequence publication, node-LDP provides a mechanism $\mathcal{M}$ that enables users to perturb their degree before sending it to an untrusted collector. By ensuring the perturbed degrees satisfy $\epsilon$-node-LDP, the collector cannot distinguish a degree from any other possible degree with high confidence.

\begin{definition}\label{def1}
\textbf{(node-LDP)} \cite{cite7}
A random algorithm $\mathcal{M}$ satisfies $\epsilon$-node-LDP, iff only any $v_i, v_j \in V$, the corresponding to adjacent bit vectors $B_i$ and $B_j$ that differ at most $n$ bits, and any output $o\in O$,

\begin{equation}  \label{eq1}
\Pr \left[ {\mathcal{M}\left( {{B_i}} \right) = o} \right] \le \exp \left( \varepsilon  \right) \times \Pr \left[ {\mathcal{M}\left( {{B_j}} \right) = o} \right]
\end{equation}

where the parameter $\epsilon$ is referred to as the privacy budget, which is used to balance the trade-off between utility and privacy in the random algorithm $\mathcal{M}$, a smaller value of $\epsilon$ implies a higher level of privacy protection. The $O$ denotes the output domain of $\mathcal{M}$.
\end{definition}

\textbf{Laplace Mechanism\cite{cite8, cite9} .} One way to satisfy differential privacy is to add noise to the output of a query. In the Laplace mechanism, in order to release $f(G)$ where $f: \mathbb{D}\to \mathbb{R}^d$ while satisfying $\epsilon$-differential privacy, one releases
\begin{equation}
  \mathcal{M}(G) = f(G) + \text{Lap}(\frac{\Delta f}{\epsilon}) \notag
\end{equation}

where $\Delta f = \underset{G\simeq G'}{\text{max}} \|f(G) - f(G')\|_1$ and $\text{Pr}[\text{Lap}(\beta) = x] = \frac{1}{2\beta}e^{-|x|/\beta}$.

\textbf{Exponential Mechanism \cite{cite23}.} This mechanism was designed for situations where we want to choose the best response, but adding noise directly to the computed quantity. Given a dataset $G$, and a utility function $u: \mathbb{G}\times O \to \mathbb{R}$, which maps dataset/output pairs to utility scores $u(G, o)$. Intuitively, for a fixed graph $G$, the user prefers that this mechanism outputs some element $o$ of $O$ with the maximum possible utility score. That is, the mechanism outputs $o$ with probability proportional to $\text{exp}(\frac{\epsilon u(G,o)}{2\Delta u})$, where $\Delta u = \underset{\forall o, G,G'}{\text{max}}|u(G, o) - u(G', o)|$ is the sensitivity of the score function. This mechanism preserves $\epsilon$-differential privacy.

\textbf{Warner Randomized Response\cite{cite21}: WRR.} In local differential privacy, this mechanism is a building block for generating a random value for a sensitive Boolean question. Specifically, each user sends the true value with probability $p$ and the false value with probability $1-p$. To adapt $\text{\tt{WRR}}$ to satisfy $\epsilon$-local differential privacy, we set $p$ as follows:

\begin{equation}
p = \frac{e^\epsilon}{e^\epsilon + 1} \notag
\end{equation}

\textbf{Composition Property\cite{cite23}.} This property guarantees that for any sequence of computations $\mathcal{M}_1, \mathcal{M}_2, ..., \mathcal{M}_k$, if each $\mathcal{M}_i$ is $\epsilon_i$-differential privacy, then releasing the results of all $\mathcal{M}_1, \mathcal{M}_2, ..., \mathcal{M}_k$ is $\sum_{i}^{k}\epsilon_i$-differential privacy.

\subsection{Secure Aggregation}
Formally, Secure Aggregation \cite{cite24} constitutes a private MPC protocol where users transmit masked inputs under additive secret sharing, ensuring the collector can only learn the aggregated sum of user data. Mask generation follows the Diffie-Hellman (denoted as DH\cite{cite25}) key exchange protocol. The DH protocol is parameterized as the triple: $DH = (KA. param, KA. gen, KA. agree)$, where $KA.param(1^\lambda)\to (Q, g, q)$ outputs a group $G$ of order $q$ with generator $g$, given a security parameter $\lambda$; $KA.gen(Q, g, q)\to (sk, pk)$ generates private key $sk = a \xleftarrow{\textdollar}\mathbb{Z}_q$ and public key $pk=g^a$; $KA.agree(sk, pk)\to k_{i,j}$ computes shared secret $k_{i,j}=pk_{j}^{sk_i}=g^{a_{i}a_j}$. $k_{i,j}$ is derived using the same public parameters $(Q, g, q)$, which ensures $k_{i,j} = k_{j,i}$. Assuming $i<j$, user $i$ adds $k_{i,j}$ to its value $x_i$, while user $j$ subtracts $k_{j, i}$ from its value $x_j$, which makes the mask eliminated in the final aggregated result. Extending this idea to $n$ users, the masking operation of user $i$ is given as follows:
\begin{equation}
m_i = \underset{1\leq i<j\leq n}{\sum}k_{i,j} - \underset{1\leq j<i\leq n}{\sum}k_{j,i}
\end{equation}

Therefore, each user performs the local encrypted computation as $y_i = x_i + m_i$, where $x_i$ is the user $i$'s original value and $m_i$ is the mask. The collector aggregates the ciphertexts from all users and computes the sum $z$ as follows:

\begin{equation}
z =  \sum_{i=1}^{n} y_i = \sum_{i=1}^{n}(x_i + \underset{1\leq i<j\leq n}{\sum}k_{i,j} - \underset{1\leq j<i\leq n}{\sum}k_{j,i}) = \sum_{i=1}^{n}x_i
\end{equation}

\section{CADR-LDP Framework}
In this section, we study the degree sequence problem in the context of node-LDP. In section 3.1, we begin by introducing our proposed framework, CADR-LDP, short for Crypto-Assisted Degree Sequence Release under node-local Differential Privacy. We then describe the algorithm for selecting the optimal parameter $\theta$ and degree order encoding mechanism in Sections 3.2 and 3.3, respectively. Section 3.4 describes the local projection method, and Section 3.5 gives the degree sequence releasing method.

\begin{algorithm}[!h]
	\caption{CADR-LDP}
	\begin{algorithmic}[1]
            \REQUIRE Set of nodes $V$, adjacency vectors of all nodes $B = \{B_1, \dots, B_n\}$, privacy budget $\varepsilon$, privacy budget allocation parameter $\alpha$, security parameter $\lambda$, minimum degree $d_{min}$ and maximum degree $d_{max}$ of graph $G$, candidate parameter domain $K$, partition size $P_{size}$.
            \ENSURE Degree sequence $DS(G')$.
            \STATE $\theta \leftarrow $\textbf{Optimal-$\theta$-Selection}$(B,\varepsilon,\lambda,d_{min},d_{max},K)$;  
            \STATE Collector computes the number of partitions $P_{num} \leftarrow \left\lceil {{{\left( {{d_{\max }} - {d_{\min }}} \right)} \mathord{\left/
            {\vphantom {{\left( {{d_{\max }} - {d_{\min }}} \right)} {{P_{size}}}}} \right.
            \kern-\nulldelimiterspace} {{P_{size}}}}} \right\rceil$;
            \STATE Collector computes partitions $P \leftarrow \{[d_{min}, d_{min} + P_{size}], \dots, [d_{min} + (P_{num} - 1)P_{size}, d_{max}]\}$;
            \STATE Collector computes the encoding range $O \leftarrow \{1, \dots, P_{num}\}$;
            \STATE Collector computes the sensitivity $\Delta u = d_{max} - d_{min}$;
            \STATE Collector shares $P, O, \Delta u$ with all users in $V$;
            \FOR{$v_i \in V$}
                \STATE User $i$ computes $o_i \leftarrow \textbf{NDOE}(B_i, \varepsilon, \alpha, P, O, \Delta u)$; // User $i$ computes its degree order $o_i$;
                \STATE User $i$ sends $o_i$ to its neighboring nodes;
            \ENDFOR
            \FOR{$v_i \in V$}
                \STATE User $i$ collects the degree orders from its neighbors;
                \STATE User $i$ performs local projection $\overline{B}_i \leftarrow \textbf{LPEA-LOW}(B_i, \varepsilon, \alpha, \theta)$;
            \ENDFOR
            \STATE $\textit{DS}(G') \leftarrow \textbf{DSR}( \{ {\overline B _1},...,{\overline B _n} \}, \varepsilon, \alpha, \theta )$;
            \RETURN $\textit{DS}(G')$.
	\end{algorithmic}
\end{algorithm}

\subsection{Overview of CADR-LDP}
In CADR-LDP, we consider three principles: (1) The local projection algorithm maximizes edge retention to reduce the releasing error; (2) The parameter selection mechanism optimizes $\theta$ with low computational and communication cost; (3) Both principles (1) and (2) must be implemented without compromising user privacy. The framework is shown in Fig. 3, which includes Optimal-$\theta$-selection, NODR, LEPA-LOW, and DSR. We provide an overview of CADR-LDP in Algorithm 1. In CADR-LDP, we propose an improved crypto-assisted solution called Optimal-$\theta$-Selection for selecting parameter $\theta$ (Section 3.2). We rely on the secure aggregation technique to design the loss function to prevent leaking the order information of individual utility loss in this solution. Before the projection, each user needs to obtain the degree order information of its neighbors. Based on OP$\epsilon$c \cite{cite26}, we propose a neighbor's degree order encoding mechanism, NDOE (Section 3.3), which can preserve the degree order while protecting users' degree counts. Then the collector depends on the partition size $P_{size}$ to divide the domain of actual degree count into $P_{num}$ partitions. Based on $P_{num}$, $d_{min}$ and $d_{max}$, the sensitivity $\Delta u$ and degree order encoding range $O$ are calculated. The collector shares the above parameters with all node users. After obtaining the degree order information using NDOE, each user sends it to its neighbors. Once each user receives the order information from its neighbors, she uses the LPEA-LOW (Section 3.4) to project her degree count. Finally, the collector uses the DSR (Section 3.5) to release the degree sequence.

\begin{figure}[htbp]
  \centering
  \includegraphics[width=\textwidth]{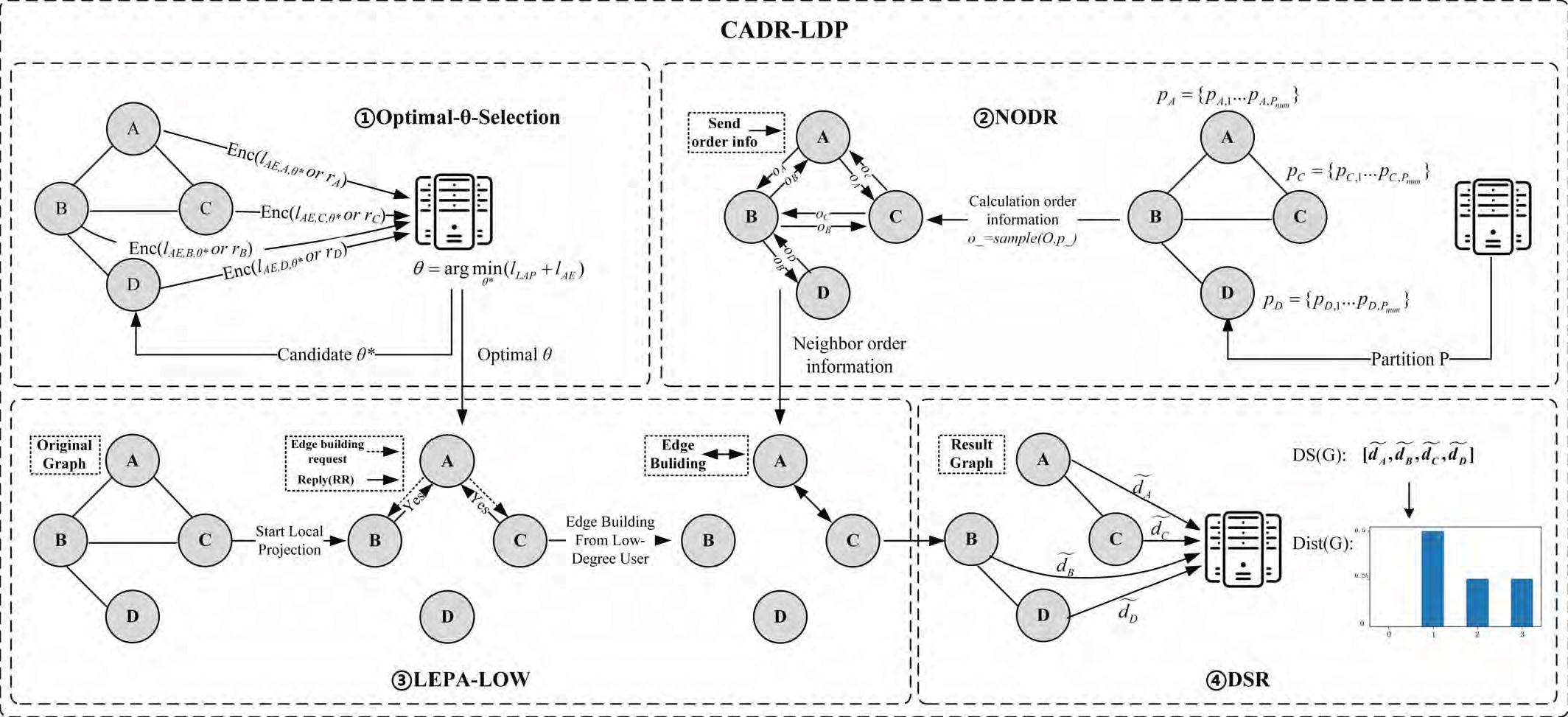}
  \caption{The framework of CADR-LDP}
\end{figure}

\subsection{Optimal Projection Parameter Selection}
As described in Algorithm 1, the local projection parameter $\theta$ directly determines the final accuracy of the degree sequence release.  A large $\theta$ value leads to excessive noise error, while a small $\theta$ value removes too many edges. Therefore, selecting the optimal $\theta$  is a significant challenge in this paper. Using the projection parameter $\theta$  to perform the edge addition process results in the Laplace noise error introduced by $\theta$  and the absolute error caused by edge addition. We formulate an optimization function incorporating both errors to derive $\theta$. Let $\mathbb{E(\cdot)}$ be the expected value. Formally, the optimization function is defined as:

\begin{equation}  \label{eq8}{L_{total}} = {l_{LAP}} + {l_{AE}} = \sum\limits_{i = 1}^n {({l_{LAP,i}} + {l_{AE,i}})}  = \sum\limits_{i = 1}^n {\left( {{\mathbb{E}}|\widetilde {{d_i}} - \overline {{d_i}} | + {\mathbb{E}}|\overline {{d_i}}  - {d_i}|} \right)}  = \sum\limits_{i = 1}^n {\left( {\frac{\theta }{\varepsilon } + |\mathop {{d_i}}\limits^ -   - {d_i}|} \right)} \end{equation}

where $L_{total}$ denotes the total error, $l_{AE}$ represents the absolute error caused by local projection, and $l_{LAP}$ denotes the noise error caused by the Laplace mechanism. We have $l_{AE} = \sum_{i = 1}^{n} l_{AE, i}$, and $l_{LAP} = \sum_{i = 1}^{n} l_{LAP, i}$, where $l_{AE, i} = \mathbb{E}|\overline{d_i}-d_i|$, $l_{LAP, i} = \mathbb{E}|\widetilde{d_i}-\overline{d_i}|$, $d_i$, $\overline{d_i}$, and $\widetilde{d_i}$ denote the original degree, the projected degree and the noise degree, respectively. According to the Laplace mechanism, we have $l_{LAP, i} = \mathbb{E}|\widetilde{d_i}-\overline{d_i}| = \frac{\theta}{\epsilon}$, which is the noise error introduced by adding Laplace noise $\text{Lap}(\frac{\theta}{\epsilon})$ to the degree $d_i$ of node $v_i$.

However, since each user lacks the global view of the entire graph, she cannot obtain the degree information of all $n$ nodes in Formula (6). This limitation prevents local computation of $\theta$ at each node side. To address this issue, we employ the collector to aggregate the local errors generated by each user and derive the global parameter $\theta$. In Formula (6), the value $l_{AE, i}$ is sensitive data, on which the collector could rely to infer the actual degree $d_i$ of node $v_i$ in terms of $d_i=l_{AE, i} + \theta$. The method described in \cite{cite19} uses Laplace noise $Lap(\frac{n-1-\theta}{\epsilon})$ with sensitivity $n$-1-$\theta$ to protect $l_{AE}$. The main limitation of \cite{cite19} is that large noise could be introduced as $n$ increases. Then, we use a crypto-assisted method based on secure aggregation to compute the optimal $\theta$. The cryptographic protocols for secure aggregation include DH key agreement and secret sharing. DH key agreement is used to generate one-time masks, and secret sharing is employed to remove masks in the case of node dropout. To simplify computation, we propose two crypto-assisted methods based on the DH key agreement protocol for solving $\theta$: Crypto-assisted-$\theta$-by Sum, and Crypto-assisted-$\theta$-by Deviation, which are summarized in Algorithm 2 and Algorithm 3.

\begin{algorithm}[H]
	\caption{Crypto-assisted -$\theta$-by Sum}
	\begin{algorithmic}[1]
            \REQUIRE Set of nodes $V$, adjacency vectors of all nodes $B = \{B_1, \dots, B_n\}$, privacy budget $\varepsilon$, security parameter $\lambda$, candidate parameter domain size $K$.
            \ENSURE Optimal projection parameter \( \theta \).
            \FOR{\( k \in [1...K] \)}
                \STATE \textbf{//User side}
                \FOR{\( v_i \in V \)}
                    \STATE User $i$ performs local projection on \( B_i \): \( {\overline {B_i}} \leftarrow \text{LPEA-LOW}(B_i, \varepsilon, \alpha, \theta) \);
                    \STATE User $i$ calculates the local projection error \( l_{AE,i,k} \leftarrow |d_i - \overline{d_i}| \); 
                    \STATE User $i$ calculates the mask value \( m_i \leftarrow \text{SA}(\lambda) \) using Equation (4);
                    \STATE  User $i$ computes the encrypted local projection error \( \text{Enc}(l_{AE,i,k}) \leftarrow l_{AE,i,k} + m_i \);  
                    \STATE  User $i$ sends \( \text{Enc}(l_{AE,i,k}) \) to the collector.
                \ENDFOR
            \ENDFOR
            \STATE \textbf{//Collector side}
            \STATE The collector calculates \( \theta \leftarrow \arg \mathop {\min }\limits_k \sum\nolimits_{i = 1}^n {\left( {\frac{k}{\varepsilon } + \text{Enc}({l_{AE,i,k}}} )\right)} \);
            \RETURN \( \theta \).
	\end{algorithmic}
\end{algorithm}

Algorithm 2 uses Equation (6) to search for the parameter $\theta$ in the range $[1...K]$. Once $\theta = k$, each user $i$ first uses the current $\theta$ to perform the local projection on $B_i$ and obtain the projected vector $\overline{B_i}$ and degree $\overline{d_i}$(Line 4). Subsequently, the user calculates the local projection error $l_{AE, i, k}$ and the mask value $m_i$ (Lines 5 and 6). Meanwhile, the user also computes the encrypted projection error $\text{Enc}(l_{AE, i, k})$ according to $m_i$ (Line 7). Then the user sends the encrypted value to the collector. Finally, the collector combines the $l_{LAP}$ and $l_{AE}$ to compute the optimal $\theta$. Although Algorithm 2 can select the optimal $\theta$, it faces high communication cost, i.e., $O(nK + d_{max}K)$. The first term results from $K$ rounds of secure aggregation where each user performs key agreement with $(n$-$1)$ other users. The second term reflects $K$ rounds of local projection, where each user may add edges for up to $d_{max}$ in each round. To address the limitation of Algorithm 2, we propose the Crypto-assisted-$\theta$-by Deviation algorithm, as shown in Algorithm 3.

According to the derivative of the error function, the communication cost of Algorithm 3 is $O(n\text{Log}_{2}K + \text{Log}_{2}K)$. The first term represents the communication cost of $\text{Log}_{2}K$ rounds of secure aggregation for key agreement, and the second term represents the communication cost of users sending information to the collector $\text{Log}_{2}K$ times. The main idea of Algorithm 3 involves differentiating the parameter $\theta$ via Equation (6) and finding the root where the derivative equals zero using the binary search method. The derivative process is shown in Equation (7).

\begin{algorithm}[H]
	\caption{Crypto-assisted -$\theta$-by Deviation}
	\begin{algorithmic}[1]
            \REQUIRE Node set $V$, adjacency vector set $B = \{B_1, \dots, B_n\}$, privacy budget $\varepsilon$, security parameter $\lambda$, candidate parameter domain size $K$.
            \ENSURE Optimal projection parameter \( \theta \).
            \STATE //\textbf{Collector side}
            \STATE The collector sets the search boundaries for \( \theta \): \( \theta_L = 1 \), \( \theta_R = K \);
            \WHILE{\( |\theta_L - \theta_R| > 1 \)}
                \STATE The collector computes \( \theta^* \leftarrow \left\lfloor {{{\left( {{\theta _L} + {\theta _R}} \right)} \mathord{\left/{\vphantom {{\left( {{\theta _L} + {\theta _R}} \right)} 2}} \right.
                 \kern-\nulldelimiterspace} 2}} \right\rfloor \);
                \STATE The collector sends \( \theta^* \) to all users;
                \STATE //\textbf{User side}
                \FOR{\( v_i \in V \)}
                    \IF{ \( d_i > \theta^* \)}
                        \STATE User $i$ sets \( r_i \leftarrow 1 \);
                    \ELSE
                        \STATE User $i$ sets \( r_i \leftarrow 0 \);
                    \ENDIF
                    \STATE User $i$ computes the mask value \( m_i \leftarrow \text{SA}(\lambda) \);
                    \STATE User $i$ computes the encrypted response \( \text{Enc}(r_i) \leftarrow r_i + m_i \);
                    \STATE User $i$ sends \( \text{Enc}(r_i) \) to the collector.
                \ENDFOR
                \STATE //\textbf{Collector side}
                \IF{\( \sum\nolimits_{i = 1}^n {\text{Enc}{(r_i)}}  < {n \mathord{\left/ {\vphantom {n \varepsilon }} \right. \kern-\nulldelimiterspace} \varepsilon } \)}
                    \STATE The collector updates \( \theta_R \leftarrow \theta^* - 1 \);
                \ELSE
                    \STATE The collector updates \( \theta_L \leftarrow \theta^* + 1 \);
                \ENDIF
            \ENDWHILE
            \RETURN \( \theta=\theta_L\).
	\end{algorithmic}
\end{algorithm}

\begin{align}\label{eq7}
\frac{\partial L_{total}}{\partial\theta} &= \frac{\partial\left(\sum_{i=1}^{n}\left(\frac{\theta}{\epsilon} + |\mathop{d_i}\limits^- -d_i|\right)\right)}{\partial\theta} \Rightarrow \frac{\partial}{\partial\theta}\left(\frac{n\theta}{\epsilon} + \sum_{i=1}^{n}\text{max}\left(0, d_i - \theta\right)\right) = 0\notag\\
&\Rightarrow  \frac{n}{\epsilon} - |\{i:d_i > \theta\}| = 0 \Rightarrow \frac{n}{\epsilon} = |\{i:d_i > \theta\}|
\end{align}

where $\frac{\partial}{\partial\theta}(\text{max}\left(0, d_i - \theta)\right)$ = -1 when $d_i>\theta$. Since the actual number of users whose degree exceed $\theta$ is unknown, we use -$|\{i:d_i>\theta\}|$ to replace $\frac{\partial}{\partial\theta}(\text{max}\left(0, d_i - \theta)\right)$, where $|\{i:d_i>\theta\}|$ denotes the number of users with $d_i>
\theta$.

According to Algorithm 3, we know that the optimal $\theta$ is the (1 - $\frac{1}{\epsilon}$) quantile of the degree sequence. The quantile guarantees that $\frac{n}{\epsilon}$ users whose degree exceeds $\theta$, minimizing the total error. Thus, Algorithm 3 leverages the collector's global view to compute the (1 - $\frac{1}{\epsilon}$) quantile. The collector performs a binary search in the candidate parameter range. In each iteration, the collector selects the median $\theta^{*}$ of this range as the candidate projection parameter, and shares $\theta^{*}$ with all users (Lines 2-5). Each user compares its degree $d_i$ with $\theta^{*}$. If $d_i > \theta^{*}$, the user sets $r$=1, and $r$=0 otherwise (Lines 6-11). However, the setting $r$ may contain sensitive information, which helps the collector infer the user's true degree. For instance, if the user $i'$ degree $d_i$=3, the collector will deduce the value by asking whether $d_i > 2$ and $d_i < 4$. To protect $r_i$, the user needs to rely on secure aggregation to encrypt it (Lines 12-14), and the collector aggregates the encrypted settings, denoted by $\sum_{i=1}^{n}\text{Enc}(r_i)$. If $\sum_{i=1}^{n}\text{Enc}(r_i)\ge\frac{n}{\epsilon}$, the collector updates the boundary $\theta_L\leftarrow\theta^* + 1$, and $\theta_L\leftarrow\theta^* - 1$ otherwise (Lines 17-21). The iteration stops until the optimal $\theta = \theta_L$.

\subsection{Neighbor's Degree Order Encoding: NDOE}


Algorithm 1 shows that node $v_i$ starts the edge addition process upon acquiring $\theta$. Fig. 2 reveals better degree sequence release accuracy when adding edges from smaller degree nodes first. However, this process requires $v_i$ to know its neighbors' degree order. Since neighbors' degree counts contain private information, these users refuse to share them explicitly with $v_i$. \cite{cite22} proposes an order-preserving encryption protocol suitable for multiple users, but it requires a trusted third party and incurs high communication costs. \cite{cite26} proposes an encoding method, OP$\varepsilon$c, which achieves partial order preservation. This method allows users to protect their data locally while achieving partial order preservation. However, this method cannot be directly applied to the node-LDP model because the method typically only protects a single value, while node-LDP requires protecting all edge information associated with a node. Therefore, based on node-LDP, we propose NODE method with OP$\varepsilon$c and the exponential mechanism, which is shown as Algorithm 4. 

\begin{algorithm}[H]
	\caption{NDOE}
	\begin{algorithmic}[1]
            \REQUIRE The degree \( d_i \) of node \( v_i \), privacy budget \( \varepsilon \), privacy budget allocation parameter \( \alpha \), partition \( P \), encoding range \( O \), and sensitivity \( \Delta u \).
            \ENSURE Degree order \( o_i \) of node \( v_i \).
            \FOR{ \( j \in O \)}
                \STATE User $i$ calculates the probability \( p_{i,j} \):
                   ${p_{i,j}} \leftarrow \frac{{{e^{ - |{d_i} - median(P[j])| \cdot \alpha \varepsilon /4\Delta u}}}}{{\sum\nolimits_{m = 1}^{|O|} {{e^{ - |{d_i} - median(P[m])| \cdot \alpha \varepsilon /4\Delta u}}} }}$
            \ENDFOR
            \STATE User $i$ calculates all sampling probability \( p_i \leftarrow \{ {p_{i,1}},...,{p_{i,|O|}}\} \);
            \STATE User $i$ samples \( o_i \) from \( O \) with probability \( p_i \) as their degree order;
            \RETURN \( o_i \).
	\end{algorithmic}
\end{algorithm}

NODE encodes the degree information of $v_i'$ neighboring nodes and uses the encoded results as their degree order information. In Algorithm 4, the user $i$ calculates the probability $p_{i,j}$ of perturbation across all intervals, where $p_{i,j}$ is the probability that node $v_i$ is perturbed into the $j$-th partition (Lines 1-4). Then, given $|O|$ positions, the user $i$ obtains the associated perturbation probability $\{p_{i,1},...,p_{i,|O|}\}$ for each location. Finally, the user $i$ samples its degree order information $o_i$ from the encoding range $O$ with probability $p_i$ (Line 5). As explained in Algorithm 1, the user $i$ uses the two parameters $P$ and $O$ shared from the collector to set the score function $u$ and the sensitivity $\Delta u$. NODE achieves the partial order preservation because it ensures that the user $i$ samples the position where its degree resides. Given the user $i'$ degree $d_i$ and the $j$-th partition, the score function $u$=-$|d_i - \text{median}(P[j])|$, where $\text{median}(P[j])$ is the median value of the $j$-th partition. 
The sensitivity \( \Delta u \) of the score function \( u \) is calculated as follows:
\begin{align}\label{eq8}
\Delta u &= \mathop {\max }\limits_{o \in O} \mathop {\max }\limits_{d_i \cong d_j} u(d_i,o) - u(d_j,o)\notag\\ &= \mathop {\max }\limits_{o \in O} \mathop {\max }\limits_{d_i \cong d_j} \left| {d_j - \text{median}(P[o])} \right| - \left| {d_i - \text{median}(P[o])} \right|
\end{align}

As can be seen from equation (8), when $d_i$=$d_{max}$ and  $d_j$=$d_{min}$, respectively, $\Delta u$=\( d_{\text{max}} - d_{\text{min}} \). Thus, the user $i$ randomly samples a value \( o_i \) from the range \( O \) with probability \( p_i \) as the degree order and sends \( o_i \) to all its neighbors.

\begin{theorem}\label{th1}
The NDOE algorithm satisfies $\frac{\alpha\varepsilon}{2}$-node-LDP.
\end{theorem}

\begin{proof}
Let \( d_i \) and \( d_j \) be the degree of $v_i$ and $v_j$ in $G$, respectively. According to the exponential mechanism, for $d_i$ and $d_j$, the NDOE algorithm outputs the degree order encoding $o$ satisfying the following equality:

\begin{equation}  \label{eq11}
\Pr \left[ {NDOE\left( {d_i,\frac{{\alpha \varepsilon }}{2}} \right) = o} \right] = \left( {\frac{{\exp (\frac{{u(d_i,o) \cdot \frac{{\alpha \varepsilon }}{2}}}{{2\Delta u}})}}{{\sum\nolimits_{o' \in O} {\exp (\frac{{u(d_i,o') \cdot \frac{{\alpha \varepsilon }}{2}}}{{2\Delta u}})} }}} \right)
\end{equation}
\begin{equation}  \label{eq12}
\Pr \left[ {NDOE\left( {d_j,\frac{{\alpha \varepsilon }}{2}} \right) = o} \right] = \left( {\frac{{\exp (\frac{{u(d_j,o) \cdot \frac{{\alpha \varepsilon }}{2}}}{{2\Delta u}})}}{{\sum\nolimits_{o' \in O} {\exp (\frac{{u(d_j,o') \cdot \frac{{\alpha \varepsilon }}{2}}}{{2\Delta u}})} }}} \right)
\end{equation}

According to Definition 2.1, 

\begin{align}\label{eq11}
\frac{\Pr\left[NODE\left(d_i, \frac{\alpha\epsilon}{2}\right) = o\right]}{\Pr\left[NODE\left(d_j, \frac{\alpha\epsilon}{2}\right) = o\right]} \notag &\le \exp\left(\frac{(u(d_i, o) - u(d_j, o))\cdot\frac{\alpha\epsilon}{2}}{2\Delta u}\right)\times\frac{\exp\left(|d_j - d_i|\right)\cdot\frac{\alpha\epsilon}{2}}{2\Delta u}\notag\\&\le \exp\left(\frac{\alpha\epsilon}{2}\right)
\end{align}
\end{proof}

Therefore, the NODE algorithm satisfies $\frac{\alpha\epsilon}{2}$ -node-LDP.

\subsection{Local Projection Edge Addition-Low: LPEA-LOW}

After obtaining the optimal projection parameter $\theta$ and the degree order of neighboring nodes, each user initiates edge-connection negotiations with its neighbors having smaller degrees. However, this mutual negotiation process inherently involves exchanging private information between parties. Therefore, we propose the LPEA-LOW with the WRR mechanism to protect the process, and the details are shown in Algorithm 5. In Algorithm 5, the user $i$ first sets an all-zero vector $\overline{B_i}$ to record the added wages (Line 2). Before the edge addition process, the user sets the position value in $\overline{B_i}$ to 1 for all established edges (Line 3). Subsequently, the user denotes the neighbors without established edges as the set $U_1$ and sends edge addition requests to them (Lines 4-5). For the user in $U_1$, if $|\overline{B_i}|<\theta$, it responds 'Yes' with $p=\frac{e^\epsilon}{e^\epsilon + 1}$ and 'No' with $q=\frac{1}{e^\epsilon + 1}$, otherwise. $p$ and $q$ are reversed (Lines 7-12). In Lines 7-12, we use the WRR mechanism to protect the privacy of the user $i'$ neighbors. If a neighboring node $v_j \in U_1$ rejects the edge addition request, it reveals that $v_j$' degree exceeds the parameter $\theta$, which is sensitive information. The user \( i \) adjusts the number of the 'Yes' responses from all neighbors, and obtains the adjusted number $c'$. The user compares $c'$ with $\theta-|\overline{B_i}|$ and obtains the minimum value $s$, that is, $s\leftarrow \text{min}(\theta-|\overline{B_i}|)$, where  $\theta-|\overline{B_i}|$ denotes the remaining edge capacity of the node $v_i$. Finally, the user selects $s$ neighbors with smaller degrees to form the set $U_3$ and establishes edges with them (Lines 15-21).

\begin{algorithm}[H]
	\caption{LPEA-LOW}
	\begin{algorithmic}[1]
            \REQUIRE Adjacency vector \( B_i \) of node \( v_i \), privacy budget \( \varepsilon \), privacy budget allocation parameter \( \alpha \), and projection parameter \( \theta \).
            \ENSURE Projected adjacency vector \( \overline{B_i} \).
            \STATE //\textbf{User side}
            \STATE User $i$ initializes an all-zero neighbor vector \(\overline{B_i} \);
            \STATE User $i$ sets the position value in \(\overline{B_i} \) to 1 for all established edges;
            \STATE User $i$ defines the set \( U_1 \) as the neighbors without established edges;
            \STATE User $i$ sends edge-addition requests to all neighbors in \( U_1 \);
            \STATE //\textbf{Neighbor side}
            \FOR{\( v_j \in U_1 \)}
                \IF{\( |{\overline B _j}| < \theta \)}
                    \STATE User $j$ responds "Yes" with probability \( p \) and "No" with probability \( q \);
                \ELSE
                    \STATE User $j$ responds "No" with probability \( p \) and "Yes" with probability \( q \);
                \ENDIF
            \ENDFOR
            \STATE //\textbf{User side}
            \STATE User $i$ collects the neighbors from \( U_1 \) who responded "Yes" and forms a set \( U_2 \);
            \STATE User $i$ calculates \( c' \leftarrow \frac{{|{U_2}| \cdot ({e^{\alpha \varepsilon /2}} + 1) - |{U_1}|}}{{{e^{\alpha \varepsilon /2}} - 1}} \);
            \STATE User $i$ calculates \( s \leftarrow \min(c', (\theta - |{\overline B _i}|)) \);
            \STATE User $i$ selects \( s \) neighbors with smaller degrees from \( U_2 \) to form set \( U_3 \);
            \FOR{\( v_j \in U_3 \)}
                \STATE User $i$ sets \( \overline{B_i}[j] \leftarrow 1 \);
                \STATE User $j$ sets \( \overline{B_j}[i] \leftarrow 1 \);
            \ENDFOR
            \RETURN \( {\overline B _i} \).
	\end{algorithmic}
\end{algorithm}

\begin{theorem}\label{th2}
Let \( |U_1| \) be the number of neighbors who responded to the edge addition request, \( c \) be the actual number of neighbors who can establish edges, and \( c' \) be the number of neighbors who can establish edges after processing by the LPEA-LOW algorithm. For any node \( v_i \), \( \text{E}(c') = c \) holds.
\end{theorem}
\begin{proof}
Let \( |U_2| \) be the number of neighbors of \( v_i \) who responded "Yes" to the edge addition request, then the following equation holds:
\begin{equation}  \label{eq13}
|{U_2}| = c \cdot \frac{{{e^\varepsilon }}}{{{e^\varepsilon } + 1}} + \left( {|{U_1}| - c} \right) \cdot \frac{1}{{{e^\varepsilon } + 1}}
\end{equation}

Since the node \( v_i \) cannot know the true value of \( c \), the corresponding estimated value \( c' \) is calculated as:
\begin{equation}  \label{eq14}
c' = \frac{{|{U_2}| \cdot ({e^\varepsilon } + 1) - |{U_1}|}}{{{e^\varepsilon } - 1}}
\end{equation}

Thus, the expected value \( \text{E}(c') \) can be expressed as:
\begin{equation}  \label{eq15}
{\rm{E}}(c') = \frac{{{\rm{E}}(|{U_2}|) \cdot ({e^\varepsilon } + 1) - |{U_1}|}}{{{e^\varepsilon } - 1}}
\end{equation}

From this equation, we conclude that \( \text{E}(c') = c \) holds.
\end{proof}

\subsection{Degree Sequence Release: DSR}
In Algorithm 1, when all users project their degrees with the parameter $\theta$, they inject the Laplace noise with the sensitivity $\theta$ into the projected degrees. They then send the noisy degree counts to the collector. Based on this idea, we propose the DSR algorithm, which is shown in Algorithm 6.


\begin{algorithm}[H]
	\caption{DSR}
	\begin{algorithmic}[1]
            \REQUIRE Set of nodes \( V \), projected adjacency vectors of all users \( \overline B  = \{ {\overline B _1},...,{\overline B _n}\} \), privacy budget \( \varepsilon \), privacy budget allocation parameter \( \alpha \), and projection parameter \( \theta \).
            \ENSURE Degree sequence \( \text{Ds}(G) \), degree distribution \( \text{Dist}(G) \).
            \STATE //\textbf{User side}
            \FOR{\( v_i \in V \)}
                \STATE User $i$ calculates the locally projected degree \( \overline{d_i} \);
                \STATE User $i$ adds Laplace noise \( \widetilde{d_i} \leftarrow \overline{d_i} + \text{Lap}(\theta/\varepsilon) \);
                \STATE User $i$ sends \( \widetilde{d_i} \) to the collector.
            \ENDFOR
            \STATE //\textbf{Collector side}
            \STATE The collector aggregates the noisy degree counts from all node users;
            \STATE The collector calculates the degree sequence \( \text{Ds}(G) \leftarrow \{ {\widetilde d_1},...,{\widetilde d_n}\} \) and the degree distribution \( \text{Dist}(G) \leftarrow \{ {p'_0},...,{p'_{n - 1}}\} \);// \( p'_i \) represents the proportion of users with noisy degree \( i \);
            \RETURN \( \text{Ds}(G), \text{Dist}(G) \).
	\end{algorithmic}
\end{algorithm}

In DSR algorithm, each user first calculates its degree $\overline{d_i}$ based on the locally projected adjacency vector $\overline{B_i}$ (Line 3). The user then employs the Laplace mechanism with sensitivity $\theta$ to disturb $\overline{d_i}$ and obtains the noise degree $\widetilde{d_i}$. Meanwhile, the user sends $\widetilde{d_i}$ to the collector (Lines 4-5). The collector aggregates the noisy degree counts and calculates the noisy degree sequence and distribution (Lines 8-9). 


\begin{theorem}\label{th3}
DSR algorithm  satisfies \( (1-\alpha)\varepsilon \)-node-LDP.
\end{theorem}

\begin{proof}
Let \( \overline{d_i} \) and \( \overline{d_j} \) be the two projected degrees of any two nodes $v_i$ and $v_j$. The noisy output $\widetilde{d}$ from DSR algorithm satisfies the following equality:

\begin{align}\label{eq15}
\frac{{\Pr [DSR(\overline d_i ,(1 - \alpha )\varepsilon ) = \widetilde d]}}{{\Pr [DSR(\overline d_j, (1 - \alpha )\varepsilon ) = \widetilde d]}} \notag &= \frac{{\exp ( - |\widetilde d - \overline d_i | \cdot \frac{{(1 - \alpha )\varepsilon }}{\theta })}}{{\exp ( - |\widetilde d - \overline d_j| \cdot \frac{{(1 - \alpha )\varepsilon }}{\theta })}} \notag \\&= \exp \left( {(|\widetilde d - \overline d_j| - |\widetilde d - \overline d_i|) \cdot \frac{(1-\alpha)\epsilon}{\theta}} \right) \notag \\ & \le \exp \left( {|\overline d_i- \overline d_j| \cdot \frac{(1-\alpha)\epsilon}{\theta} } \right) \le \exp ((1 - \alpha )\varepsilon )
\end{align}
\end{proof}

According to Definition 2.1, DSR algorithm satisfies $(1-\alpha)\epsilon$-node-LDP.

\begin{theorem}\label{th4}
The CADR-LDP algorithm satisfies \( \varepsilon \)-node-LDP.
\end{theorem}

\begin{proof}
In CADR-LDP algorithm, NODE algorithm uses the exponential mechanism with privacy budget $\frac{\alpha\epsilon}{2}$. LPEA-LOW algorithm uses WRR mechanism with privacy budget $\frac{\alpha\epsilon}{2}$. And DSR algorithm uses the Laplace mechanism with privacy budget $(1-\alpha)\epsilon$. All other steps are independent and performed on the noisy output. By composition property\cite{cite23}, the CADR-LDP algorithm satisfies \( \varepsilon \)-node-LDP.
\end{proof}

\section{Experiment}
In this section, we report experimental results comparing our
proposed solutions with the state of the art, and analyzing how different aspects of our proposed solutions affect the utility.

\subsection{Datasets and Settings}
Our experiments run in Python on a client with Intel core i5-7300HQ CPU, 16GB RAM running Windows 10. We use 8 real-world graph datasets from SNAP (https://snap.stanford.edu/data/), as shown in Table 1. The datasets are from different domains, including citation, email, and social networks. We preprocessed all graph data to be undirected and symmetric. Table 1 also shows some additional information such as $|V|$, $|E|$, $d_{max}$, and $d_{avg}$, where $|V|$ denotes the number of nodes, $|E|$ is the number of edges, $d_{max}$ denotes maximum degree, and $d_{avg}$ denotes average degree. We compare our methods, $\text{\tt{LPEA-LOW}}$ and $\text{\tt{LPEA-HIGH}}$, against $\text{\tt{RANDOM-ADD}}$ and $\text{\tt{EDGE-REMOVE}}$ for publishing node degree sequence while satisfying node-LDP. $\text{\tt{LPEA-HIGH}}$ is a variant of $\text{\tt{LPEA-LOW}}$, which prioritizes nodes with higher degrees when adding edges. $\text{\tt{RANDOM-ADD}}$ algorithm uses the edge-addition process to add edges among nodes randomly. $\text{\tt{EDGE-REMOVE}}$ algorithm uses edge-deletion process to delete edges among nodes. The four projection methods are compared in both non-private and private scenarios, and all results are average values over 20 runs. 

\begin{table}[H]
    \centering
    \caption{Characteristics of Datasets}
    \begin{tabular}{ccccc}
         \hline 
         \textbf{Graph}& \textbf{$|V|$}& \textbf{$|E|$} & \textbf{$d_{max}$}& \textbf{$d_{avg}$} \\
         \hline
        Facebook&4039&88234&1045&43.69 \\ \hline
        Wiki-Vote&7115&100762&1065&	28.32 \\ \hline
        CA-HepPh&12008&118505&491	&19.74 \\ \hline
        Cit-HepPh&34546&420899&846	&24.37 \\ \hline
        Email-Enron&36692&183831&1383	&10.02 \\ \hline
        Loc-Brightkite&58228&214078&1134&	7.35 \\ \hline
        Twitter&81306&1342303&3383	&33.02 \\ \hline
        Com-dblp&317080&1049866&343	&6.62 \\ \hline
    \end{tabular}    
    \label{tab:1}
\end{table}

\subsection{Performance Comparison in Non-Private Setting}
Based on the above 8 datasets, we first directly compare our proposed projection method, $\text{\tt{LPEA-LOW}}$ (abbreviated as $\textbf{LL}$), with three other projection methods mentioned in Section 4.1:  $\textbf{ER}$ ($\text{\tt{EDGE-REMOVE}}$), $\textbf{RA}$ ($\text{\tt{RANDOM-ADD}}$), and $\textbf{LH}$ ($\text{\tt{LPEA-HIGH}}$). We use two metrics: the MAE error and $\frac{|E'|}{|E|}$, where $E'$ denotes the edges after projection. Each dataset in Table 2 contains 3 rows: the first row shows the MAE of the releasing degree sequence, the second row shows the MAE of the degree distribution, and the third row shows the ratio of $\frac{|E'|}{|E|}$. We consider three $\theta$ values, $\theta$ = 16, 64, 128.

\begin{table}[H]
\centering
\caption{Comparison of Four Projection Methods}
\begin{tabular}{ccccccccccccc}
\hline
{\textbf{Graph}} & \multicolumn{4}{c}{\textbf{$\theta=16$}} & \multicolumn{4}{c}{\textbf{$\theta=64$}} & \multicolumn{4}{c}{\textbf{$\theta=128$}} \\  \cline{2-13} & \textbf{ER} & \textbf{RA} & \textbf{LH} & \textbf{LL} & \textbf{ER} & \textbf{RA} & \textbf{LH} & \textbf{LL} & \textbf{ER} & \textbf{RA} & \textbf{LH} & \textbf{LL}  \\
\hline
{\textbf{Facebook}}  &34.56&31.98&32.58&\textbf{31.02}&16.17&14.90&15.92&\textbf{13.38}&5.80&5.79&6.34&\textbf{4.71} \\&
1.27&1.27&1.28&\textbf{1.27}&\textbf{0.49}&0.50&0.50&0.52&\textbf{0.26}&0.27&0.28&0.27\\
&0.21&0.27&0.25&\textbf{0.29}&0.63&0.66&0.64&\textbf{0.69}&0.87&0.87&0.85&\textbf{0.89} \\
\hline
{\textbf{Wiki-Vote}}  &24.31&23.04&23.40&\textbf{21.46}&4.73&14.05&14.94&\textbf{12.23}&8.01&7.82&8.58&\textbf{6.82} \\
&0.95&1.12&1.18&\textbf{0.68}&0.40&0.56&0.59&\textbf{0.31}&0.22&0.31&0.34&\textbf{0.17} \\
&0.14&0.19&0.17&\textbf{0.24}&0.48&0.50&0.47&\textbf{0.57}&0.72&0.72&0.70&\textbf{0.76} \\
\hline
{\textbf{CA-HepPh}}  &13.87&13.32&13.45	&\textbf{12.67}	&7.71&	7.44&	7.61&	\textbf{6.58}&	4.75&	4.54&	4.80&	\textbf{3.80} \\
&0.49&0.53&0.58	&\textbf{0.47}	&0.18	&0.19	&0.22&	\textbf{0.17}&	0.12&	0.13&	0.14&	\textbf{0.09} \\
&0.30&0.33&0.33	&\textbf{0.35}	&0.61	&0.62&	0.61&	\textbf{0.67}&	0.76&	0.77&	0.76&	\textbf{0.81} \\
\hline
{\textbf{Cit-HepPh}}  &15.70&15.45&15.78&\textbf{13.56}	&5.10&	5.35&	5.68&	\textbf{4.41}&	1.72&	1.77&	1.79&	\textbf{1.62} \\
&0.94&0.98&1.02&\textbf{0.94}	&0.19	&0.25&	0.29&	\textbf{0.18}&	0.08&	0.08&	0.10&	\textbf{0.06} \\
&0.36&0.36&0.35&\textbf{0.44}	&0.79	&0.78&	0.77&	\textbf{0.82}&	0.93&	0.93&	0.93&	\textbf{0.93} \\
\hline
{\textbf{Email-Enron}}  
&6.71&6.71&6.77&\textbf{6.20}	&4.07	&4.25&	4.35&	\textbf{3.72}&	2.60&	2.76&	2.85&	\textbf{2.42} \\
&0.62&0.75&0.77&\textbf{0.55}	&0.41	&0.52&	0.54&	\textbf{0.28}&	0.30&	0.42&	0.45&	\textbf{0.18} \\
&0.32&0.33&0.32&\textbf{0.38}	&0.59	&0.57&	0.56&	\textbf{0.63}&	0.74&	0.72&	0.71&	0.76\\
\hline
{\textbf{Loc-Brightkite}}  &3.71&3.74&3.81&\textbf{3.30}	&1.51&	1.58	&1.65&	\textbf{1.34}&	0.71&	0.75&	0.76&	0\textbf{.65} \\
&0.27&0.42&0.44&\textbf{0.22}&	0.08	&0.13&	0.14&	\textbf{0.05}&	0.04&	0.06&	0.07&	\textbf{0.02} \\
&0.49&0.49&0.48&\textbf{0.55}	&0.79&	0.78	&0.78&	\textbf{0.82}&	0.90&	0.90&	0.89&	\textbf{0.91} \\
\hline
{\textbf{Twitter}}  &26.14&25.07&25.73&\textbf{23.66}&14.78	&14.53	&15.59&	\textbf{12.97}&	8.62&	8.66&	9.39&	\textbf{7.70} \\
&0.93&0.96&1.00&\textbf{0.93}	&0.31	&0.38&	0.46&	\textbf{0.28}&	0.16&	0.19&	0.25&	\textbf{0.13} \\
&0.21&0.24&0.22&\textbf{0.28}&	0.55	&0.56&	0.53&	\textbf{0.61}&	0.73&	0.73&	0.71&	\textbf{0.77} \\
\hline
{\textbf{Com-dblp}}  
&1.92&1.98&2.07&\textbf{1.67}&	0.26	&0.26&	0.27&	\textbf{0.22}&	0.034&	0.036&	0.036&	\textbf{0.030} \\
&0.23&0.31&0.37&\textbf{0.18}&	0.02	&0.03&	0.04&	\textbf{0.01}&	0.005&	0.005&	0.006	&\textbf{0.004} \\
&0.71&0.70&0.69&\textbf{0.75}&	0.96&	0.96	&0.96&	\textbf{0.97}&	0.99&	0.99	&0.99&	\textbf{0.99} \\
\hline      
\end{tabular}
\end{table}

The results demonstrate that the three edge-addition methods ($\textbf{LH}$, $\textbf{LL}$, and $\textbf{RA}$) have lower MAE error and higher $\frac{|E'|}{|E|}$ for releasing degree sequence, meaning that these three methods maintain enough edges after projection. However, $\text{\tt{LPEA-HIGH}}$ and $\text{\tt{RANDOM-ADD}}$ have higher degree distribution publication error than $\text{\tt{EDGE-REMOVE}}$. This is because their edge-addition processes prioritize nodes with higher degrees, neglecting to consider those with smaller degrees. On the other hand, $\text{\tt{LPEA-LOW}}$ preserves the most number of edges and achieves the lowest error in degree sequence releasing and degree distribution under different projection parameter $\theta$. The main reason is that $\text{\tt{LPEA-LOW}}$ starts edge addition from nodes with smaller degrees.


\subsection{Optimal Projection Parameter Selection}

To explore the selection of the optimal parameter $\theta$, we evaluate the performance of Algorithm 2 and Algorithm 3 with $\epsilon = \{1.0, 1.5, 2.0, 2.5, 3.0\}$. Table 3 shows the results. From Table 3, we can see that Algorithm 3 obtains the optimal $\theta$ when $\epsilon = 3.0$, which minimizes the MAE of the degree sequence release. Algorithm 3, however, may fail to guarantee finding the optimal $\theta$ that minimizes the total error, but rather approximates the optimal solution. This is because the degree loss caused by projection equals $|d_i - \theta|$ for users with degree $d_i \ge \theta$, and users with $d_i < \theta$ incur zero degree loss. The true degree loss from projection should be $d_i - \overline{d_i}$ as defined in Algorithm 2. Consequently, Algorithm 2 can identify the projection parameter that minimizes total error. However, obtaining $\overline{d_i}$ requires users to perform an additional round of local projection, resulting in the communication overhead.


\begin{table}
\centering
\caption{Optimal projection parameter $\theta$}
\begin{tabular}{cccccc}
\hline
\multicolumn{1}{c}{\textbf{Graph}} & $\varepsilon=1$ & $\varepsilon=1.5$ & $\varepsilon=2$ & $\varepsilon=2.5$ & $\varepsilon=3$ \\
\hline
\multicolumn{1}{c}{\textbf{Facebook}} & 1 & 15 & 25 & 34 & 42 \\
\multicolumn{1}{c}{\textbf{Wiki-Vote}} & 1 & 2 & 4 & 10 & 16 \\
\multicolumn{1}{c}{\textbf{CA-HepPh}} & 1 & 3 & 5 & 7 & 9 \\
\multicolumn{1}{c}{\textbf{Cit-HepPh}} & 1 & 9 & 15 & 20 & 24 \\
\multicolumn{1}{c}{\textbf{Email-Enron}} & 1 & 2 & 3 & 4 & 5 \\
\multicolumn{1}{c}{\textbf{Loc-Brightkite}} & 1 & 1 & 2 & 3 & 4 \\
\multicolumn{1}{c}{\textbf{Twitter}} & 1 & 8 & 15 & 21 & 27 \\
\multicolumn{1}{c}{\textbf{Com-dblp}} & 1 & 3 & 4 & 5 & 5 \\
\hline            
\end{tabular}
\end{table}

\subsection{Performance Comparison in Private Setting}
In the private setting, Figures 4-6 show the degree sequence and distribution release for $\text{\tt{LPEA-LOW}}$, $\text{\tt{LPEA-HIGH}}$, $\text{\tt{RANDOM-ADD}}$ and $\text{\tt{EDGE-REMOVE}}$. We first apply $\theta =  \{1, 2, \dots, 100\}$ and $\epsilon = 3$ on the above four methods and report the experimental result in Fig. 4. The result shows that the MAE and MSE of degree sequence releasing exhibit a convex trend as $\theta$ varies. This trend occurs because the $l_{AE}$ error caused by projection decreases monotonically with $\theta$ increasing, while the $l_{LAP}$ error from Laplace noise increases as $\theta$ varies. Thus, the two errors create an optimization trade-off that determines the optimal $\theta$ value.
From the results in Figure 4, we can see that $\text{\tt{LPEA-LOW}}$ consistently achieves lower MAE and MSE across the 8 datasets compared to the other three algorithms. This is because $\text{\tt{LPEA-LOW}}$ starts edge addition from neighboring nodes with smaller degrees, which is particularly effective in sparse graphs (e.g., Loc-Brightkite and Com-dblp). 
The MAE and MSE of degree sequence release of all the methods with the optimal $\theta$ from Table 3 are depicted in Fig. 5 while $\epsilon$ varies from 1 to 3. The results show that the MAE and MSE drop when $\epsilon$ decreases. The proposed solution $\text{\tt{LPEA-LOW}}$ achieves good accuracy over all datasets and outperforms all the other differential private methods. In the Loc-Brightkite dataset, for example, when the privacy budget is relatively large, e.g., $\epsilon=3.0$, its MAE and MSE always stay below or close to 6.1 and 40.0, respectively. When $\epsilon$ decreases, the accuracy drops but its MAE is still smaller than 11.0 even when $\epsilon=1.0$. The other methods do not perform as well as our proposed method because when adding edges, we consider the willingness of neighboring users whose degrees are relatively smaller.

Fig. 6 compares the quality of the resulting degree distributions of $\text{\tt{LPEA-LOW}}$ to those of $\text{\tt{LPEA-HIGH}}$, $\text{\tt{RANDOM-ADD}}$, and $\text{\tt{EDGE-REMOVE}}$ when setting \( \alpha = 0.1 \), \( \varepsilon = 3 \), and the optimal \( \theta \) from Table 3. The $\text{\tt{ORIGIN-GRAPH}}$ line represents the true degree distribution of the original graph. The experimental result shows that $\text{\tt{EDGE-REMOVE}}$ generally performs the worst, follow by $\text{\tt{RANDOM-ADD}}$ and $\text{\tt{LPEA-HIGH}}$. Looking at intermediate results, we note that $\text{\tt{EDGE-REMOVE}}$ is far from $\text{\tt{ORIGIN-GRAPH}}$. Our proposed method outperforms three other methods significantly, yielding a quite accurate degree distribution, especially for larger datasets. 


\begin{figure}[htbp]
  \centering
  \includegraphics[width=\textwidth]{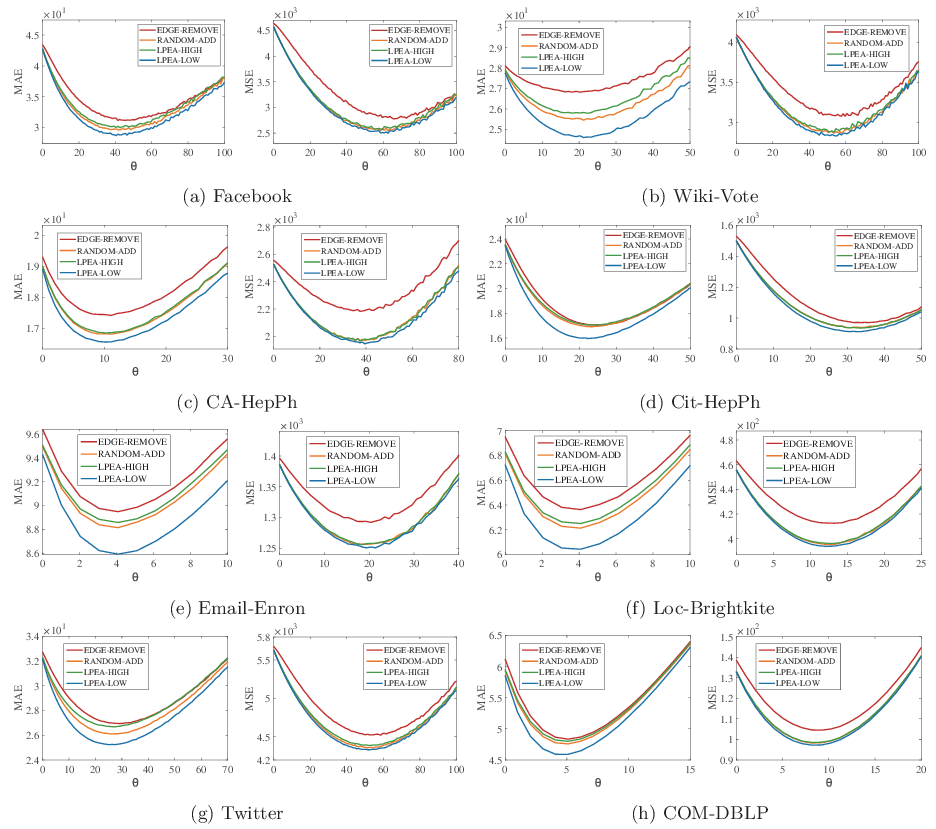}
  \caption{The MAE and MSE of algorithms on different datasets at $\epsilon=3.0$, varying $\theta$}
\end{figure}

\begin{figure}[htbp]
  \centering
  \includegraphics[width=\textwidth]{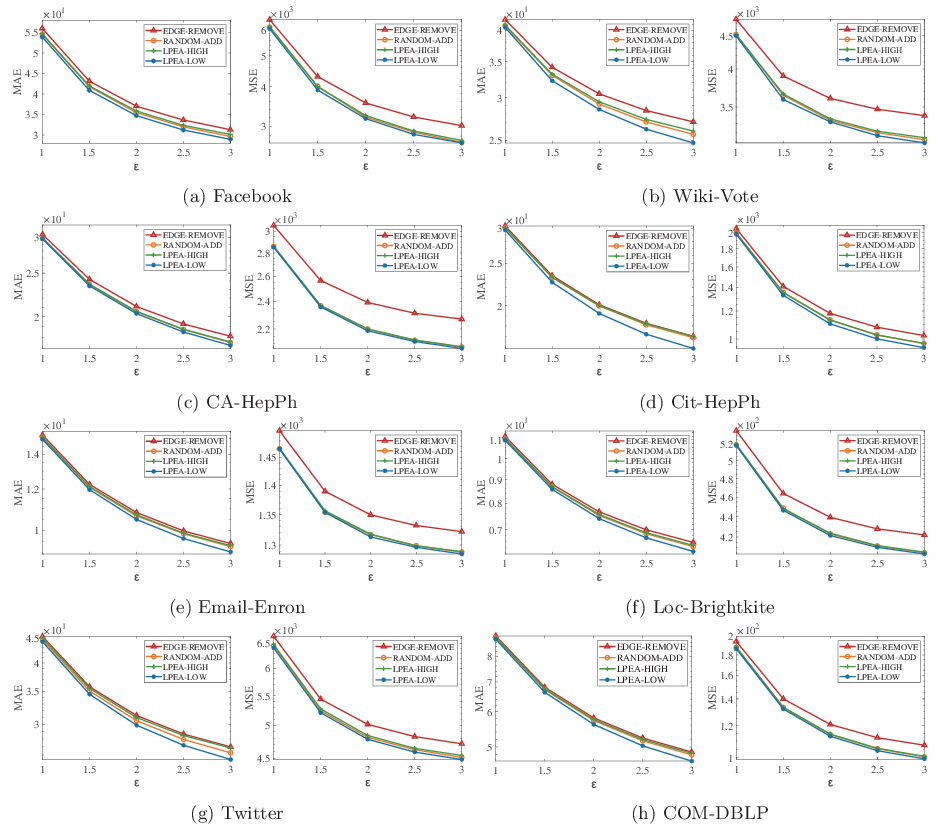}
  \caption{The MAE and MSE of algorithms on different datasets, varying $\epsilon$ }
\end{figure}

\begin{figure}[htbp]
  \centering
  \includegraphics[width=\textwidth]{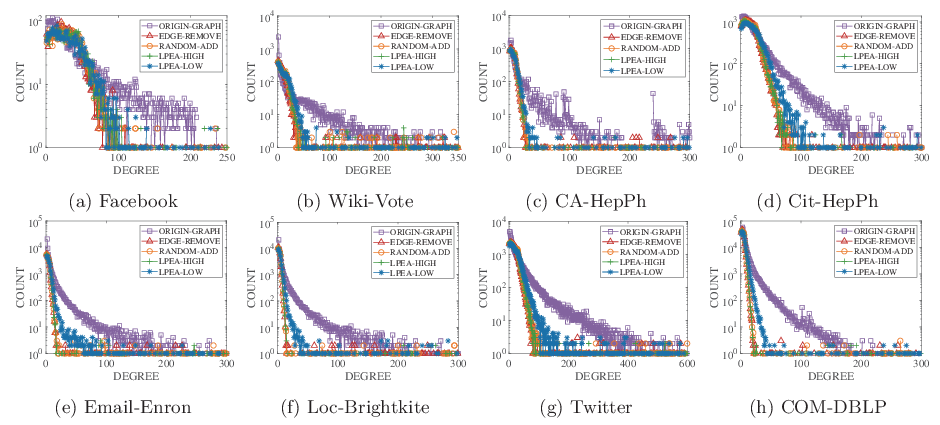}
  \caption{The comparison of degree distribution of algorithms on different datasets}
\end{figure}

\section{Related Work}
Applying node central differential privacy (node-CDP) to the degree distribution and degree sequence has been studied extensively. We argue that using node-CDP is more meaningful than edge-CDP, as node-CDP provides personal privacy control on graph data. The work in \cite{cite27} studies the problem of releasing degree sequence under edge-CDP, which employs the Laplace mechanism directly to generate $\text{Lap}\frac{2}{\epsilon}$ noise and performs a consistency inference to boost the accuracy of the noise sequence. Several subsequent works \cite{cite28, cite29} have explored releasing the degree sequence or distribution. However, these approaches become infeasible under node-CDP because the sensitivity may be unbounded. The techniques of using graph projection to bound the sensitivity are studied to release the degree distribution with node-CDP \cite{cite2, cite18, cite30}. The edge addition process, introduced in \cite{cite2}, randomly inserts each edge correlated to a node with a degree exceeding $\theta$. In contrast, the truncation approach \cite{cite18} eliminates all nodes with degree above $\theta$. Meanwhile, the edge deletion process \cite{cite30} iterates through all edges in random order and removes those linked to nodes with degrees surpassing $\theta$. The existing projection methods, however, are designed only for node-CDP and are not applicable under node-LDP and edge-LDP.

In the local setting, existing works almost use edge-LDP to study various graph statistics, such as degree distribution \cite{cite7}, subgraph counting \cite{cite10, cite11, cite12, cite13, cite14, cite15, cite16, cite17}, and synthetic graph generation \cite{cite6, cite31}. For example, the work in \cite{cite7} proposes the one-round method $\text{\tt{LF-GDRR}}$ with $\text{\tt{WRR}}$ \cite{cite21} mechanism to reconstruct the graph structure for answering statistical queries. Based on $\text{\tt{LF-GDRR}}$, $\text{\tt{Local2Rounds}}$ and $\text{\tt{ARRTwoNS}}$ methods are proposed in \cite{cite10, cite11}, respectively. The former reduces the triangle count error by introducing a download strategy, and the latter enhances the accuracy and reduces the communication cost of the former method by incorporating the Asymmetric Random Response ($\text{\tt{ARR}}$) and a conditional download strategy. Subsequently, $\text{\tt{WS}}$ \cite{cite12} solution is proposed for answering triangle and 4-cycle counting queries. This approach employs a wedge shuffle technique to achieve privacy amplification while reducing the error compared to $\text{\tt{Local2Rounds}}$ and $\text{\tt{ARRTwoNS}}$. Compared to $\text{\tt{Local2Rounds}}$, $\text{\tt{TCRR}}$ \cite{cite13} approach relies on a different postprocessing of $\text{\tt{WRR}}$ to answer triangle count queries, which can obtain tight upper and lower bounds on the variance of the resulting. Unlike \cite{cite10, cite11, cite12, cite13}, $\text{\tt{PRIVET}}$ \cite{cite14} introduces a federated triangle counting estimator that preserves privacy through edge relationship-LDP, specifically designed for correlated data collection scenarios. Meanwhile, $\text{\tt{DBE}}$ \cite{cite15} specializes in butterfly counting for bipartite graphs with edge-LDP, and $\text{\tt{OddCycleC}}$ \cite{cite16} targets cycle counting in degeneracy-bounded graphs. In contrast to the aforementioned works in the local model, several studies focus on generating synthetic graphs under edge-LDP. $\text{\tt{LDPGen}}$ \cite{cite6} proposes a novel multi-phase approach to synthetic
decentralized social graph generation, and $\text{\tt{sgLDP}}$ \cite{cite31} carefully estimates the joint distribution of attributed graph data under edge-LDP, while preserving general graph properties such as the degree distribution, community structure, and attribute distribution. 

Despite the above method using edge-LDP with lower sensitivity to protect graph statistics, they have the following limitations regarding degree sequence release: (1) edge-LDP provides a weaker privacy guarantee than node-LDP, and (2) each user in these methods needs to send an adjacency vector of length $n$, resulting in high communication costs. To address these limitations, $\text{\tt{ EDGE-REMOVE}}$ \cite{cite19, cite20} uses node-LDP with   encryption techniques to release the degree sequence and distribution, which employs an edge deletion process to project the degree of each node in the original graph. A key contribution of this work lies in its dual capability to simultaneously provide higher release accuracy while enhancing privacy guarantees.  However, from Fig. 2, we note that $\text{\tt{EDGE-REMOVE}}$ may lead to more edges deleted than necessary due to ignoring the willingness of each user. Moreover, the crypto-assisted parameter selection process requires each user to perform a round of local projection for each iteration, resulting in high communication costs. In comparison to existing works, our method $\text{\tt{CADR-LDP}}$ introduces an edge addition process with  encryption techniques, which further achieves better utility and satisfies node-LDP.

\section{Conclusion}
This work focused on preserving node privacy for releasing degree sequence and distribution under node-LDP with cryptographic techniques. To strike a better trade-off between utility and privacy, we first proposed an edge addition process, which locally bounded the degree of each user with the optimal projection parameter $\theta$. To obtain  $\theta$, two estimation methods were further proposed, where we introduce  encryption techniques and the exponential mechanism to provide the privacy guarantee. Finally, the extensive experiments on real-world datasets showed that our solution outperformed its competitors.

\begin{acks}
The research reported in this paper was supported by the National Natural Science Foundation of China (No. 62072156), the Natural Science Outstanding Youth Science Foundation Project of Henan Province (252300421061), and the Basic Research Special Projects of Key Research Projects in Higher Education Institutes in Henan Province (25ZX012). We appreciate the anonymous reviewers for their helpful
comments on the manuscript.
\end{acks}

\bibliographystyle{ACM-Reference-Format}
\bibliography{sample-base}

\appendix

\end{document}